\documentclass[english,runningheads]{article}
\usepackage{latexsym}
\usepackage{graphicx,amssymb,amsmath,amssymb,mathrsfs}
\usepackage{hyperref}
\usepackage{cite}
\usepackage[linesnumbered, vlined, ruled]{algorithm2e}
\usepackage{amsthm}
\usepackage{fullpage}
\usepackage[nodayofweek]{datetime}
\usepackage{enumerate}
\usepackage{multirow}
\usepackage[mathlines]{lineno}
\usepackage{color}

\graphicspath{{figs/}}
\newcommand{\ShoLong}[2]{#2} 

\newcommand{\spm}{\ensuremath{S\!P\!M}}
\newcommand{\spt}{\ensuremath{S\!P\!T}}

\newcommand{\Real}{\ensuremath{\mathbb{R}}}
\newcommand{\Plane}{\ensuremath{\mathbb{R}^2}}
\newcommand{\bd}{\ensuremath{\partial}}

\newcommand{\intr}{\ensuremath{\mathrm{int}}}

\newcommand{\seg}{\overline}

\newcommand{\far}{\ensuremath{\phi}}
\newcommand{\geoball}{\ensuremath{\mathbf{B}}}
\newcommand{\geocen}{\ensuremath{\mathrm{cen}}}
\newcommand{\rad}{\ensuremath{\mathrm{rad}}}
\newcommand{\diam}{\ensuremath{\mathrm{diam}}}

\newcommand{\complain}[1]{\marginpar{\framebox{~~~}}{\textcolor{red}{#1}}}

\def\denseitems{
    \itemsep1pt plus1pt minus1pt
    \parsep0pt plus0pt
    \parskip0pt\topsep0pt}

\begin{document}


\title{Computing the $L_1$ Geodesic Diameter and Center of a Simple Polygon in Linear Time\thanks{A preliminary version of this paper appeared in the proceedings of the {\em11th Latin American Theoretical INformatics Symposium} (LATIN'14)~\cite{bkow-clgdcsplt-13}.}}
\author{
Sang Won Bae\thanks{Kyonggi University, Suwon, South Korea.{\tt swbae@kgu.ac.kr}
}
\and
Matias Korman\thanks{National Institute of Informatics, Tokyo, Japan. {\tt korman@nii.ac.jp}} $^,$\thanks{JST, ERATO, Kawarabayashi Large Graph Project.}
\and
Yoshio Okamoto\thanks{The University of Electro-Communications, Tokyo, Japan. {\tt okamotoy@uec.ac.jp}
}
\and
Haitao Wang\thanks{Utah State University, Logan, Utah, USA. {\tt haitao.wang@usu.edu}
}
}

\maketitle

\newtheorem{theorem}{Theorem}
\newtheorem{lemma}[theorem]{Lemma}
\newtheorem{corollary}[theorem]{Corollary}
\newtheorem{fact}[theorem]{Fact}
\newtheorem{observation}[theorem]{Observation}



\begin{abstract}
In this paper, we show that the $L_1$ geodesic diameter and center
of a simple polygon can be computed in linear time.
For the purpose, we focus on revealing basic geometric properties of the $L_1$ geodesic balls,
that is, the metric balls with respect to the $L_1$ geodesic distance.
More specifically, in this paper we show that
any family of $L_1$ geodesic balls in any simple polygon has Helly number two,
and the $L_1$ geodesic center consists of midpoints of shortest paths between diametral pairs.
These properties are crucial for our linear-time algorithms,
and do not hold for the Euclidean case.
\end{abstract}

\section{Introduction}
\label{sec:intro}

Let $P$ be a simple polygon with $n$ vertices in the plane.
The \emph{diameter} and \emph{radius} of $P$ with respect to a certain metric $d$
are very natural and important measures of $P$. The diameter with respect to $d$ is defined to be the maximum distance over all pairs of points in $P$, that is, $\max_{p,q \in P} d(p, q)$,
while the radius is defined to be the min-max value $\min_{p\in P} \max_{q\in P} d(p, q)$. Here, the polygon $P$ is considered as a closed and bounded space
and thus the diameter and radius of $P$ with respect to $d$ are well defined.
A pair of points in $P$ realizing the diameter is called
a \emph{diametral pair}. Similarly, any point $c$ such that
$\max_{q\in P} d(c, q)$ is equal to the radius
is called a {\em center}. In this paper we study how fast can we compute these measures (and whenever possible, to also obtain the set of points that realize them).

%

One of the most natural metrics on a simple polygon $P$
is induced by the length of the Euclidean shortest paths that stay within $P$,
namely, the \emph{(Euclidean) geodesic distance}.
The problem of computing the diameter and radius of a simple polygon with respect to the geodesic distance has been intensively studied in computational geometry since the early 1980s.
The diameter problem was first studied by Chazelle~\cite{c-tpca-82}, who gave an $O(n^2)$-time algorithm. The running time was afterwards improved to $O(n\log n)$ by Suri~\cite{s-cgfnsp-89}.
Finally, Hershberger and Suri~\cite{hs-msspm-97} presented a linear-time algorithm
based on a fast matrix search technique.
Recently, Bae et al.~\cite{bko-gdpd-13} considered the diameter problem
for polygons with holes.

The first algorithm for finding the Euclidean geodesic radius was
given by Asano and Toussaint~\cite{at-cgcsp-85}. In their study, they
showed that any simple polygon has a unique center, and provided an
$O(n^4\log n)$-time algorithm for computing it. The running time was
afterwards reduced to $O(n\log n)$ by Pollack, Sharir, and
Rote~\cite{psr-cgcsp-89}. Since then, it has been a longstanding open
problem whether the center can be computed in linear time (as also
mentioned later by Mitchell~\cite{m-spn-04}).

Another popular metric with a different flavor is the \emph{link distance},
which measures the smallest possible number of links (or turns) of piecewise linear paths.
The currently best algorithms that compute the link diameter or radius
run in $O(n \log n)$ time~\cite{suri-mlpprp-87,k-ealdp-89,dls-aclcsp-92}.
The \emph{rectilinear link distance} measures the minimum number of links
when feasible paths in $P$ are constrained to be rectilinear.
Nilsson and Schuierer~\cite{ns-crldp-91,ns-oarlcrp-96} showed how to solve the problem under the rectilinear link distance in linear time.

In order to tackle the open problem of computing the Euclidean geodesic center, we investigate another natural metric: the  $L_1$ metric.
To the best of our knowledge, only a special
case where the input polygon is rectilinear has been
considered in the literature. This result is given by Schuierer~\cite{s-cl1dcsrp-94}, where he showed how to compute the $L_1$ geodesic diameter and radius of a simple rectilinear polygon in linear time.

This paper aims to provide a clear and complete exposition
on the diameter and radius of general simple polygons with
respect to the $L_1$ geodesic distance. We first focus on
revealing basic geometric properties of the geodesic
balls (that is, the metric balls with respect to the $L_1$
geodesic distance). Among other results, we show that any family of $L_1$ geodesic balls has Helly number
two (see Theorem~\ref{thm:geoball_helly}). 
This is a crucial property that does not hold for the Euclidean geodesic distance, and thus
we identify that the main difficulty of the open problem lies there.


We then show that the method of Hershberger and Suri~\cite{hs-msspm-97} for computing the Euclidean diameter extends to $L_1$ metrics, and that the running time is preserved. However, the algorithms for computing the Euclidean radius do not  easily extend to rectilinear metrics. 
Indeed, even though the approach of Pollack et al.~\cite{psr-cgcsp-89} can be adapted for  the $L_1$ metric, the running time is $O(n \log n)$. On the other hand, the algorithm of Schuierer~\cite{s-cl1dcsrp-94} for rectilinear simple polygons heavily exploits properties derived from rectilinearity. Thus, its extension to general simple polygons is not straightforward either.

In this paper we use a different approach: using our Helly-type theorem for $L_1$ geodesic balls, we show that the set of points realizing $L_1$ geodesic centers coincides with the intersection of  a finite number of geodesic balls. Afterwards we show how to compute their intersection in linear time.
Table~\ref{table:summary} summarizes the currently best results
on computing the diameter and radius of a simple polygon with respect to
the most common metrics, including our new results.

\begin{table}[t]
\label{table:summary} \centering \caption{Summary of currently best
results on computing the diameter and radius of a simple polygon $P$
with respect to various metrics on $P$. } 
\begin{tabular}{cllcccc }
\hline
\multicolumn{2}{c}{Metric} & Restriction on $P$\phantom{M} & \multicolumn{2}{c}{Diameter} & \multicolumn{2}{c}{Radius} \\
\hline
\hline
 \multirow{3}*{Geodesic\phantom{M}} & Euclidean\phantom{MM} & simple & $O(n)$ & \cite{hs-msspm-97} & $O(n \log n)$ & \cite{psr-cgcsp-89} \\
\cline{2-7}
  & \multirow{2}*{$L_1$} & rect.\ simple  & $O(n)$ & \cite{s-cl1dcsrp-94} & $O(n)$ & \cite{s-cl1dcsrp-94} \\
  &  & simple  & $O(n)$ & [Thm.~\ref{thm:diameter}] & $O(n)$ & [Thm.~\ref{thm:center}] \\
\hline
\multirow{2}*{Link} & regular & simple & $O(n\log n)$  & \cite{suri-mlpprp-87} & $O(n\log n)$ & \cite{k-ealdp-89,dls-aclcsp-92} \\
\cline{2-7}
 & rectilinear & rect.\ simple & $O(n)$ & \cite{ns-crldp-91} & $O(n)$ & \cite{ns-oarlcrp-96}\\
\hline
\end{tabular}
\vspace{-6pt}
\end{table}

\ShoLong{Due to page limit, several proofs are omitted. They can be found in the extended version of this paper~\cite{bkow-clgdcsplt-13}.}{}

\section{Preliminaries}
\label{sec:pre}

For any subset $A\subset \Real^2$, we denote by $\bd A$ and $\intr A$
the boundary and the interior of $A$, respectively.
For  $p,q \in \Real^2$,
denote by $\seg{pq}$ the line segment with endpoints $p$ and $q$.
For any path $\pi$ in $\Plane$,
let $|\pi|$ be the length of $\pi$ under the $L_1$ metric,
or simply the \emph{$L_1$ length}.
Note that $|\seg{pq}|$ equals the $L_1$ distance between $p$ and $q$.

The following is a basic observation on the $L_1$ length of paths in $\Real^2$.
A path is called \emph{monotone} if any vertical or horizontal line intersects it
in at most one connected component.
\begin{fact} \label{fact:l1length}
 For any monotone path $\pi$ between $p,q\in\Real^2$,
 it holds that $|\pi| = |\seg{pq}|$. 
\end{fact}

Let $P$ be a simple polygon with $n$ vertices.
We regard $P$ as a compact set in $\Plane$, so its boundary $\bd P$ is contained in $P$.
An \emph{$L_1$ shortest path} between $p$ and $q$ is a path
joining $p$ and $q$ that lies in $P$ and minimizes its $L_1$ length.
The \emph{$L_1$ geodesic distance} $d(p,q)$ is the $L_1$ length
of an $L_1$ shortest path between $p$ and $q$.
We are interested in two quantities: the \emph{$L_1$ geodesic diameter} $\diam(P)$
and \emph{radius} $\rad(P)$ of $P$, defined to be
$ \diam(P) := \max_{p,q\in P} d(p,q)$ and $\rad(P) := \min_{p\in P} \max_{q\in P} d(p,q)$.
Any pair of points $p, q\in P$ such that $d(p,q) = \diam(P)$ is called
a \emph{diametral pair}.
A point $c\in P$ is said to be an \emph{$L_1$ geodesic center} if and
only if $\max_{q\in P} d(c, q) = \rad(P)$. We denote by $\geocen(P)$
the set containing all $L_1$ geodesic centers of $P$.

Analogously, a path lying in $P$ minimizing its \emph{Euclidean} length
is called the \emph{Euclidean shortest path}.
It is well known that there is always a unique Euclidean shortest path
between any two points in a simple polygon~\cite{ghlst-ltavspptsp-87}.
We let $\pi_2(p,q)$ be the unique Euclidean shortest path
from $p\in P$ to $q\in P$.
The following states a crucial relation between Euclidean and $L_1$ shortest paths
in a simple polygon.
\begin{fact}[Hershberger and Snoeyink~\cite{hs-cmlpghc-94}] \label{fact:euc_simple}
 For any two points $p,q\in P$,
 the Euclidean shortest path $\pi_2(p,q)$ is also an $L_1$ shortest path between $p$ and $q$.
\end{fact}
Notice that this does not imply coincidence between
the Euclidean and the $L_1$ geodesic diameters or centers,
as the lengths of paths are measured differently\ShoLong{. }{ (see an example in Figure~\ref{fig_examples}).
\begin{figure}[tb]
  \center
  \includegraphics[width=.60\textwidth]{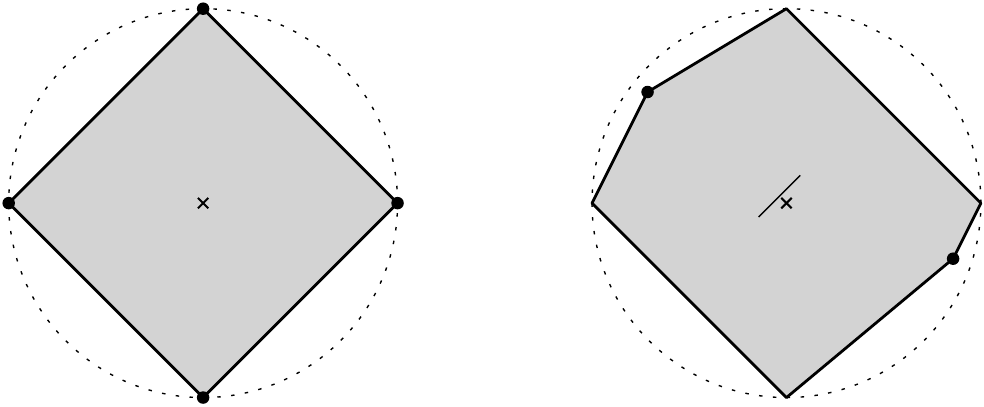}
  \caption{Examples of simple polygons and their center/diameters under both $L_1$ and Euclidean metrics. The left figure is the unit $L_1$ ball centered at the origin (the Euclidean ball is also shown for clarity). In the $L_2$ metric the diameter is realized by any opposite corners of the polygon. The diameter is equal in $L_1$ metric, but now it is realised by any two points in opposite edges of the boundary. In this problem instance the radius is realized by the same point in both metrics (the origin, depicted as a cross). In the right figure the unit ball has been slightly modified with two additional vertices (of coordinates $(-\frac{5}{7},\frac{4}{7})$ and $(\frac{6}{7},-\frac{2}{7})$). Since the polygon is still contained in the Euclidean unit ball, neither the Euclidean diameter or center change with the addition of the new vertices. However, the $L_1$ diameter is now uniquely realized by the two new vertices. Also, the set of points that realize the radius has shifted slightly upwards and is no longer unique: any point in the segment of endpoints $-(\frac{2}{14},\frac{1}{14})$ and $(\frac{1}{14},\frac{2}{14})$ (the thin segment in the Figure) is an $L_1$ center.}
  \label{fig_examples}
\end{figure}
}
Nonetheless, Fact~\ref{fact:euc_simple} enables us
to exploit structures for Euclidean shortest paths such as the shortest path map.

A \emph{shortest path map} for a source point $s\in P$
is a subdivision of $P$ into regions according to the combinatorial structure of shortest paths
from $s$.
For Euclidean shortest paths,
Guibas et al.~\cite{ghlst-ltavspptsp-87} showed that
the shortest path map $\spm(s)$ can be computed in $O(n)$ time.
Once we have $\spm(s)$, the Euclidean geodesic distance from $s$ to any query point $q\in P$
can be computed in $O(\log n)$ time, and the actual path $\pi_2(s, q)$ in additional time
proportional to the complexity of $\pi_2(s, q)$.
Fact~\ref{fact:euc_simple} implies that the map $\spm(s)$ also plays
a role as a shortest path map for the $L_1$ geodesic distance so that
a query $q\in P$ can be processed in the same time bound
to evaluate the $L_1$ geodesic distance $d(s, q)$ or to obtain the shortest path $\pi_2(s, q)$.



Throughout the paper, unless otherwise stated,
$P$ refers to a simple polygon,
a shortest path and the geodesic distance always refer to an $L_1$ shortest path
and the $L_1$ geodesic distance $d$,
and the geodesic diameter/center is always assumed to be with respect to the $L_1$ geodesic distance $d$.


\section{The $L_1$ Geodesic Balls}
\label{sec:geoball}
Geodesic balls (or geodesic disks) are metric balls under the geodesic distance $d$.
More precisely, the \emph{$L_1$ (closed) geodesic ball} centered at point $s\in P$
with radius $r \in \Real$, denoted by $\geoball_s(r)$,
is the set of points $x\in P$ such that $d(s, x) \leq r$.
Note that if $r < 0$, it holds that $\geoball_s(r) = \emptyset$.
In this section, we reveal several geometric properties of the geodesic balls $\geoball_s(r)$,
which build a basis for our further discussion.

\subsection{$P$-convex sets}
A subset $A \subseteq P$ is \emph{$P$-convex} if
for any $p, q\in A$, the Euclidean shortest path $\pi_2(p,q)$
is a subset of $A$.
The $P$-convex sets are also known as the \emph{geodesically convex} sets
in the literature~\cite{s-cgfnsp-89}.
Pollack et al.~\cite{psr-cgcsp-89} achieved their $O(n \log n)$-time
algorithm computing the Euclidean geodesic center based on
the $P$-convexity of Euclidean geodesic balls. A set $A$ is {\em
path-connected} if and only if, for any $x,y\in A$, there exists a path $\pi$ connecting them such that $\pi\subseteq A$. With this definition we can introduce an equivalent condition of $P$-convexity.
\begin{lemma} \label{lem:P-convex}
For any subset $A\subseteq P$ of $P$, the following are equivalent.
 \begin{enumerate}[(i)] \denseitems
 \item $A$ is $P$-convex.
 \item $A$ is path-connected and for any line segment $\ell\subset P$, $A\cap \ell$ is connected.
\end{enumerate}
\end{lemma}
\newcommand{\pfLemPConv}{
\begin{proof}
\textit{(i) $\Rightarrow$ (ii).}
Assume that $A$ is $P$-convex.
By definition, any two points $p,q \in A$ are joined by the Euclidean shortest path
$\pi_2(p,q)$ which lies in $A$, that is, $A$ should be path-connected.
Now, suppose that there exists a line segment $\ell\subset P$ such that $A\cap \ell$
consists of two or more connected components.
Then, there are two points $p, q\in A\cap \ell$ such that
the segment $\seg{pq}$ is not completely contained in $A$.
Since $\seg{pq}\subseteq \ell \subset P$, we have $\pi_2(p, q) = \seg{pq}$.
This contradicts the $P$-convexity of $A$.

\begin{figure}[tb]
  \center
  \includegraphics[width=.60\textwidth]{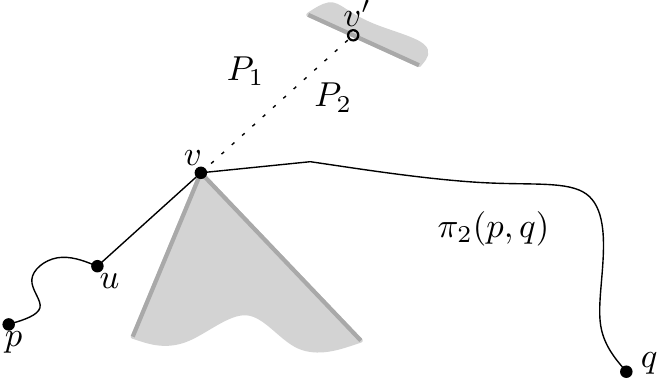}
  \caption{Illustration of the proof of Lemma~\ref{lem:P-convex}:
    the region shaded by gray color is outside $P$, that is, $\Plane \setminus P$. }
  \label{fig:pconvex}
\end{figure}

\textit{(ii) $\Rightarrow$ (i).}
Assume that $A$ is path-connected and any line segment intersects $A$ in a connected subset.
For a contradiction, suppose that there exist $p, q\in A$ such that $\pi_2(p,q)\not \subseteq A$.
Recall that $\pi_2(p,q)$ is a sequence of line segments
whose endpoints are either vertices of $P$, $p$, or $q$~\cite{ghlst-ltavspptsp-87}.
In particular, there exists a vertex of $P$ that is in $\pi_2(p,q)$ but not in $ A$. Let $v$ be the first such vertex, and let $u$ be the vertex of $\pi_2(p,q)$ just before $v$ (note that $u$ is either $p$ or a corner of $P$). By construction, $u$ and $v$ are adjacent through a segment of $\pi_2(p, q)$ and $v$ is farther from $p$ than $u$. By our choice of $v$, we must have $u\in A$ (see \figurename~\ref{fig:pconvex} for an illustration).

Let $v' \in \bd P$ be the first point on $\bd P$ hit by the ray starting at $v$
in the direction opposite to $u$.
Since $\seg{vv'}$ is a diagonal of $P$ (that is, $v, v'\in\bd P$ and $\seg{vv'} \subset P$),
$P$ is partitioned by $\seg{vv'}$ into two simple polygons. Let $P_1$ be the sub-polygon containing $p \in P_1$; we claim that $q$ must be in the other one (which we call $P_2$). Observe that $\pi_2(p,q)$ properly crosses $\seg{vv'}$ at $v$. Indeed, if $q\in P_1$, the path $\pi_2(p,q)$ must properly cross again the segment $\seg{vv'}$. However, in such a case we can find a shorter path that avoids the segment $\seg{vv'}$ altogether, a contradiction. Thus, we conclude that $p\in P_1$ and $q\in P_2$.


Since $A$ is path-connected, there exists a path $\pi$ between $p$ and $q$ such that $\pi \subseteq A$. Since $P$ is a simple polygon,
the path $\pi$ must cross $\seg{vv'}$, and thus $\seg{vv'} \cap A \neq \emptyset$.
Now, consider the segment $\seg{uv'}$.
As observed above, we have $u\in A$, $v\notin A$, and $\seg{vv'} \cap A \neq \emptyset$. In particular, $A\cap \seg{uv'}$ is not connected, a contradiction.
\end{proof}
}
\ShoLong{}{\pfLemPConv}

We are interested in the boundary of a $P$-convex set.
Let $A \subseteq P$ be a $P$-convex set.
Consider any convex subset $Q \subseteq P$.
Since $\pi_2(p,q) = \seg{pq}$ for any $p, q\in Q$,
the intersection $A \cap Q$ is also a convex set due to the $P$-convexity of $A$.
Based on this observation, we show the following lemma.
\begin{lemma} \label{lem:pconvex-bdconvex}
 Let $A\subseteq P$ be a closed $P$-convex set.
 Then, any connected component $C$ of $\bd A \cap \intr P$
 is a convex curve.
\end{lemma}
\newcommand{\pfLemBCConvex}{
\begin{proof}
 Consider an arbitrary triangulation $\mathcal{T}$ of the simple polygon $P$.
 We call each edge in $\mathcal{T}$ a \emph{diagonal} if it is not an edge of $P$.
 For each triangle $\triangle \in \mathcal{T}$, $A \cap \triangle$ is a convex set
 since $A$ is $P$-convex and $\triangle$ is convex.
 Therefore, $\bd (A \cap \triangle)$ is convex, implying that every (maximal) continuous portion of $\bd A \cap \triangle$ is a convex curve (note that $\bd A \cap \triangle$ may have more than one continuous portion).
 Now, consider any connected component $C$ of $\bd A \cap \intr P$.
 The diagonals of $\mathcal{T}$ subdivides $C$ into some pieces,
 each of which is a convex curve, as discussed above.
 Since the vertices of each triangle in $\mathcal{T}$ are those of $P$,
 $C$ does not pass through any vertex of any triangle by definition.
 This implies that $C$ is a simple curve, being either open or closed.

 In order to show that $C$ is a convex curve,
 it suffices to check each intersection point $p$ between $C$ and a triangle diagonal $e$.
 If such a point $p$ does not exist, then $C$ does not intersect any diagonal, and we are done.
 Otherwise, the diagonal $e$ is incident to two triangles $\triangle_1$ and $\triangle_2$ of $\mathcal{T}$.  Let $\alpha_1$ be the continuous portion of $C \cap \triangle_1$ that connects $p$ and let $\alpha_2$ be the continuous portion of $C \cap \triangle_2$ that connects $p$. According to our discussion above, both $\alpha_1$ and $\alpha_2$ are nonempty convex curves.
 Let $U$ be a sufficiently small disk centered at $p$ such that $(\alpha_1\cup\alpha_2)\cap U$ is part of the boundary of $A\cap U$. By the $P$-convexity of $A$, $A \cap U$ must also be a convex set,
 and therefore, $(\alpha_1\cup\alpha_2)\cap U$, as a part of $\bd (A \cap U)$, is a convex curve. This also shows that $\alpha_1\cup \alpha_2$ is convex.
 Since there are only a finite number of such intersection points $p$,
 repeating this argument for each such $p$ proves that $C$ is a convex curve.
\end{proof}
}

\ShoLong{}{\pfLemBCConvex}

Note that if a connected component $C$ of $\bd A \cap \intr P$ is
not a closed curve, then $C$ is an open curve excluding its endpoints,
which lie on $\bd P$.
This implies that the curve $C$ divides $P$ into two connected components
such that $\intr A$ lies on one side of $C$, regardless of whether $C$ is open or closed.

\subsection{Geometric properties of $L_1$ geodesic balls}
In the following, we establish several geometric properties of geodesic balls $\geoball_s(r)$ that follow from the $P$-convexity of $\geoball_s(r)$.
Note that, to the best of our knowledge, most of these properties of $\geoball_s(r)$ have not been discussed before in the literature.

We start with a simple observation.
By Fact~\ref{fact:euc_simple}, $\pi_2(s, p)$ is an $L_1$ shortest path from $s$ to $p\in P$.
Since $\pi_2(s, p)$ makes turns only at vertices of $P$,
the ball is equal to the union of some $L_1$ balls centered at the vertices of $P$.
More precisely, $\geoball_s(r) = \bigcup_{v \in V \cup \{s\}} \geoball_v(r - d(s,v))$,
where $V$ denotes the set of vertices of $P$.
This immediately implies the following observation.
\begin{observation} \label{obs:geoball1}
 For any $s\in P$ and $r > 0$,
 the geodesic ball $\geoball_s(r)$ is a simple polygon in $P$
 and each side of $\geoball_s(r)$ either lies on $\bd P$ or has slope $1$ or $-1$.
\end{observation}

\begin{figure}[tb]
  \center
  \includegraphics[width=\textwidth]{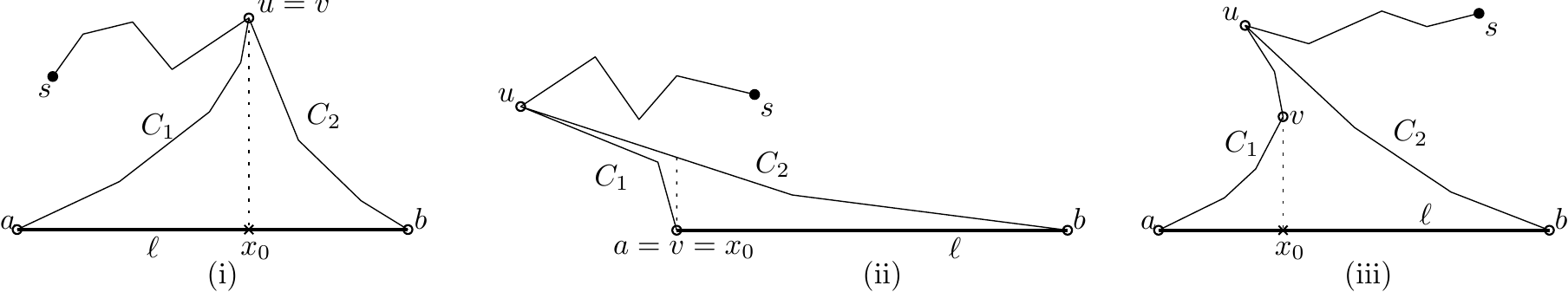}
  \caption{Illustration of the proof of Lemma~\ref{lem:dist_seg_horizontal}:
    (i) When both chains $C_1$ and $C_2$ of the funnel are monotone and $v$ is visible to $\ell$.
    (ii) When both chains $C_1$ and $C_2$ of the funnel are monotone and $v$ is to the left of $\ell$.
    (iii) When $C_1$ is not monotone. }
  \label{fig:dist_seg_horizontal}
\end{figure}

\begin{lemma}  \label{lem:dist_seg_horizontal}
 Given a point $s\in P$ and a horizontal or vertical line segment $\ell \subset P$,
 the function $f(x) = d(s, x)$ over $x\in \ell$ is convex.
\end{lemma}
\begin{proof}
Without loss of generality, we assume that $\ell$ is horizontal.
The case where $\ell$ is vertical can be handled in a symmetric way.
Consider the union of all Euclidean shortest paths $\pi_2(s, x)$ from $s$ to $x$ over $x\in \ell$,
which forms a \emph{funnel} $F$ with apex $u$ and base $\ell$ plus $\pi_2(s,u)$.
The funnel $F$ consists of two concave chains $C_1$ and $C_2$ through vertices of $P$ and
the endpoints of $\ell$
so that $C_1$ connects the apex $u$ and the left endpoint of $\ell$ and
$C_2$ connects $u$ and the right endpoint of $\ell$.
See \figurename~\ref{fig:dist_seg_horizontal}.
Note that the apex $u$ is also a vertex of $P$ unless $u=s$.

Each of the two concave chains $C_1$ and $C_2$ is
either monotone or not.
Recall that a path is called monotone if and only if any vertical or horizontal line intersects
it at most once.
Observe that at least one of them must be monotone,
since the apex $u$ must be visible to a point of $\ell$.
Without loss of generality, we assume that
$C_2$ is monotone.
Let $a$ and $b$ be the left and right endpoints of $\ell$, respectively.
For any point $p$, let $p_x$ denote its $x$-coordinate.
We define a particular vertex $v$ of $F$ defined as follows:
(1) $v = u$ if both chains are monotone and $a_x\leq u_x\leq b_x$ (see
\figurename~\ref{fig:dist_seg_horizontal} (i));
(2) $v = a$ if both chains are monotone and $u_x< a_x$ (see
\figurename~\ref{fig:dist_seg_horizontal} (ii));
(3) $v = b$ if both chains are monotone and $b_x< u_x$ (symmetric to \figurename~\ref{fig:dist_seg_horizontal} (ii));
(4) $v$ is the rightmost vertex of $C_1$ if $C_1$ is not monotone (see
\figurename~\ref{fig:dist_seg_horizontal} (iii)).

Note that, by construction of $v$, it always holds that $a_x\leq v_x\leq
b_x$. Let $x_0 \in \ell$ be the point of $\ell$ that has the same $x$-coordinate as $v$.
We first prove that $v$ is visible to $x_0$. In case (1) both
chains $C_1$ and $C_2$ are monotone and $\ell$ is horizontal. Thus, the segment connecting $v$ and $x_0$ cannot be obstructed. In cases (2) and (3) we have $x_0=v$ (and thus the statement is
obviously true). For case (4)  the vertex $v$ cuts $C_1$ into two monotone concave chains.
Since $\ell$ is horizontal and $u$ is visible to at least one point of $\ell$,
$v$ must be visible to $x_0$.

We claim that for any $x\in \ell$ there exists a shortest path to $u$ that passes through $x_0$. Consider the path formed by concatenating the (horizontal) segment $\seg{xx_0}$, the (vertical) segment $\seg{x_0v}$ and the shortest path from $v$ to $u$. Recall that we considered four different types of funnels. In most of these cases, the resulting path is monotone (and thus it is a shortest path by Fact~\ref{fact:l1length}). The only situation in which the resulting path is not monotone is if the funnel is in case (4), and $x$ is to the left of $x_0$. However, in this situation we know that $\pi_2(s,x)$ always passes through $v$.

That is, for any $x\in \ell$ there exists a shortest path to $u$ that passes through $x_0$ (and thus to $s$ as well, since all shortest paths from points of $\ell$ pass through $u$). In particular, we have $d(s,x)=|\pi_2(s,x_0)| + |\seg{x_0x}|$ which is a convex function over $x\in \ell$ as claimed.
\end{proof}

We are ready to prove the $P$-convexity of any $L_1$ geodesic ball.
\begin{lemma} \label{lem:geoball_pconvex}
 For any point $s\in P$ and any real $r\in \Real$,
 the $L_1$ geodesic ball $\geoball_s(r)$ is $P$-convex.
\end{lemma}
\begin{proof}
The case where $r\leq 0$ is trivial, so assume $r>0$. Suppose that $\geoball_s(r)$ is not $P$-convex.
Since $\geoball_s(r)$ is a simple polygon (Observation~\ref{obs:geoball1}),
any line segment in $P$ intersects $\geoball_s(r)$ in finitely many connected components.
Thus, by Lemma~\ref{lem:P-convex},
there exists a line segment $\ell \subset P$ such that
$\ell$ crosses $\bd \geoball_s(r) \cap \intr P$ exactly twice.
Let $a, b\in \ell$ be the two intersection points such that
$\seg{ab} \setminus\{a, b\}$ lies in $P\setminus \geoball_s(r)$.

\begin{figure}[t]
  \center
  \includegraphics[width=.75\textwidth]{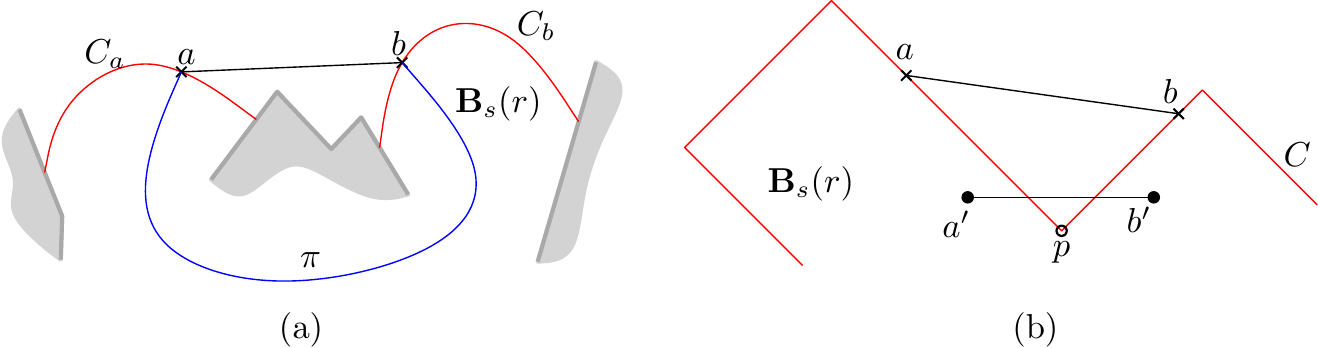}
  \caption{Illustration of the proof of Lemma~\ref{lem:geoball_pconvex}.
  The region shaded by gray depicts $\Plane \setminus P$.}
  \label{fig:geoball_pconvex}
  \vspace{-10pt}
\end{figure}

We then observe that $a$ and $b$ belong to a common connected component $C$
of $\bd \geoball_s(r) \cap \intr P$.
Suppose for a contradiction that $a$ and $b$ belong to different components $C_a$ and $C_b$,
respectively.
See \figurename~\ref{fig:geoball_pconvex}(a).
Since $\geoball_s(r)$ is path-connected and closed,
there exists a path $\pi$ between $a$ and $b$ such that $\pi\subset \geoball_s(r)$.
Consider the simple closed curve $L:=\seg{ab}\cup\pi$, and $L$ is contained in $P$.
Since $C_a \neq C_b$ and $a\in C_a$, $L$ separates the two endpoints of $C_a$,
that is, an endpoint of $C_a$ (which is on $\partial P$) lies in the region bounded by $L$.
However, this is impossible since $P$ is simple, a contradiction.
Hence, both $a$ and $b$ lie in a common connected component $C$
of $\bd \geoball_s(r) \cap \intr P$.

By Observation~\ref{obs:geoball1}, $C$ is a polygonal curve consisting of line segments
with slope $1$ or $-1$.
Since $a, b \in C$ and $\seg{ab}$ is not contained in $\geoball_s(r)$,
$C$ has a reflex corner $p$ incident to two line segments 
whose slopes are $1$ and $-1$, respectively.
See \figurename~\ref{fig:geoball_pconvex}(b) for an illustration.
Then, we can find a horizontal or vertical line segment $\seg{a'b'}$ sufficiently close to $p$
such that $a', b'\in \geoball_s(r)$ and $\geoball_s(r) \cap \seg{a'b'}$
consists of two connected components.
Take any point $x\in \seg{a'b'} \setminus \geoball_s(r)$.
Since $a', b' \in \geoball_s(r)$ but $x \notin \geoball_s(r)$,
we have a strict inequality $d(s, x) > r \geq d(s, a')$ and $d(s,x) > r \geq d(s, b')$,
a contradiction to Lemma~\ref{lem:dist_seg_horizontal}.
\end{proof}
%
%

The $P$-convexity of the geodesic balls,
together with Lemma~\ref{lem:pconvex-bdconvex} and Observation~\ref{obs:geoball1},
immediately implies the following corollary.
\begin{corollary} \label{coro:geoball_pconvex1}
 For $s\in P$ and $r>0$, each connected component $C$ of $\bd \geoball_s(r) \cap \intr P$
 is a convex polygonal curve consisting of line segments of slope $1$ or $-1$.
\end{corollary}

The following corollary can also be easily derived from Lemmas~\ref{lem:P-convex} and \ref{lem:geoball_pconvex}.
\begin{corollary} \label{coro:geoball_seg}
 For any point $s\in P$ and any $r>0$,
 the geodesic ball $\geoball_s(r)$ intersects any line segment in $P$
 in a connected subset.
\end{corollary}
\newcommand{\pfCorGeoBall}{
}
\ShoLong{}{\pfCorGeoBall}

A real-valued function $f$ is called \emph{quasiconvex} if its sublevel set
$\{x \mid f(x) \leq a\}$ for any $a\in \Real$ is convex.
Corollary~\ref{coro:geoball_seg} implies the following.
\begin{corollary} \label{coro:dist_seg}
 Given a point $s\in P$ and a line segment $\ell \subset P$,
 the function $f(x) = d(s, x)$ over $x\in \ell$ is quasiconvex.
\end{corollary}
\newcommand{\pfCorDistSeg}{
\begin{proof}
For any $r>0$, $\geoball_s(r) \cap \ell$ is connected by Corollary~\ref{coro:geoball_seg},
and is therefore convex.
It suffices to observe that $\geoball_s(r) \cap \ell$ coincides with
the sublevel set of $f(x)$ over $x \in \ell$ at $r$;
that is, $\geoball_s(r) \cap \ell = \{x\in \ell \mid d(s,x) \leq r\}$.
\end{proof}
}
\ShoLong{}{\pfCorDistSeg}

Indeed, the geodesic distance function $d(s, x)$ over $x\in \ell$ is
not only quasiconvex but convex;
this can be shown by a more careful geometric analysis.
Nonetheless, the quasiconvexity will be sufficient for our overall purpose.

\subsection{Helly-type theorem for geodesic balls}
Here, we discuss the intersection of a family of $L_1$ geodesic balls,
and show that the $L_1$ geodesic balls have Helly number two.
More precisely, we claim the following theorem.
\begin{theorem} \label{thm:geoball_helly}
 Let $\mathcal{B}$ be a family of closed $L_1$ geodesic balls.
 If the intersection of every two members of $\mathcal{B}$ is nonempty,
 then 
 $\bigcap \{ B \mid B \in \mathcal{B} \} \neq \emptyset$.
\end{theorem}

In the following, we prove Theorem~\ref{thm:geoball_helly}.
For the purpose, we make use of a Helly-type theorem on simple polygons
proven by Breen~\cite{b-httsp-96}.
\begin{theorem}[Breen~\cite{b-httsp-96}] \label{thm:helly_simple}
 Let $\mathcal{P}$ be a family of simple polygons in the plane.
 If every three (not necessarily distinct) members of $\mathcal{P}$ have
 a simply connected union and
 every two members of $\mathcal{P}$ have a nonempty intersection,
 then $\bigcap\{P \mid P \in \mathcal{P}\} \neq \emptyset$.
\end{theorem}

Thus, it suffices to show that the union of two or three balls
is simply connected, provided that any two of them have a nonempty intersection.
This can be done based on the above discussion on the geodesic balls
with Lemma~\ref{lem:pconvex-bdconvex} and Corollary~\ref{coro:geoball_pconvex1}.
\begin{lemma} \label{lem:threeballs}
 Let $B_1, B_2, B_3$ be any three closed $L_1$ geodesic balls
 such that every two of them have a nonempty intersection.
 Then, the union $B_1 \cup B_2 \cup B_3$ is simply connected.
\end{lemma}
\begin{proof}
Since every pair of balls have nonempty intersection, the union $B_1 \cup B_2 \cup B_3$ is connected. Assume to the contrary that the union $B_1 \cup B_2 \cup B_3$ has a hole $H$. This hole is a simple polygon whose boundary $\bd H$ consists of
portions of $\bd B_i \cap \intr P$ and $\bd P$ for $i=1, 2, 3$.
We consider each of the connected components $C_1, C_2, \ldots, C_m$ of $\bd B_i \cap \intr P$ and $\bd P$ that appear on $\bd H$.


We claim that, for any $i=1, 2, 3$, only a single connected component of $\bd B_i \cap \intr P$ can appear on $\bd H$. Assume otherwise that there are two components, say $C_1$ and $C_2$, that satisfy $C_1, C_2 \subset \bd B_1 \cap \intr P$. By definition of connected component, $C_1$ partitions $P$ into two components. 
Since $C_1$ appears on $\bd B_1$, $H$ lies in one side of $C_1$ and $B_1$ in the other side $C_2$.
Analogously, $H$ lies in one side of $C_2$ and $B_1$ in the other side of $C_2$.
Pick any two points $p_1 \in C_1 \cap \bd H$ and $p_2 \in C_2 \cap \bd H$. Since $H$ is a simple polygon (and thus connected), there exists a path $\pi$ between $p_1$ and $p_2$ that stays inside $H$ (and in particular avoids $B_1$). On the other hand, since $B_1$ is an $L_1$ geodesic ball, there exists a path $\pi'$ between $p_1$ and $p_2$ that stays inside $B_1$. Now consider the loop $L:=\pi \cup \pi'$ constructed by merging the two paths $\pi$ and $\pi'$. By construction, an endpoint of $C_1$ must be contained in the region bounded by $L$, which is a contradiction since the endpoints of $C_1$ by definition must belong to $\bd P$.



By a similar reasoning we claim that $\bd H$ cannot contain any portion of $\bd P$ (or equivalently $\bd H \subset \intr P$). Indeed, assume for the sake of contradiction that there exists a point $p \in \bd P \cap \bd H$. Consider a sufficiently small neighborhood around $p$: $\bd P$ partitions this neighborhood into two components, one of which must be contained in $H$. Recall that by hypothesis we have $\bd H \subset \bd (B_1 \cup B_2 \cup B_3)$, and in particular the other component of the neighborhood must be contained in $B_1 \cup B_2 \cup B_3$. However, both regions are subsets of $P$ and as such they must both lie in the same side, a contradiction. 

Thus, the boundary of $H$ is formed by portions of at most one connected component of each ball. Recall that, by Corollary~\ref{coro:geoball_pconvex1}, each connected component is a convex polygonal curve consisting of line segments of slope $1$ or $-1$. In particular, $\bd H$ also consists of line segments of slope $\pm 1$ and thus must contain at least four line segments. Pick any point $p\in \intr H$ and shoot four rays of slope $\pm 1$ until they hit a point of $\bd H$. Since each connected component is convex, we conclude that each ray must hit a different component. In particular, $H$ must be formed by at least four connected components, a contradiction. 

\end{proof}

\ShoLong{}{Therefore, Theorem~\ref{thm:geoball_helly} follows directly from
Theorem~\ref{thm:helly_simple} and Lemma~\ref{lem:threeballs}.
We note that Theorem~\ref{thm:geoball_helly}
does not hold for the Euclidean case. It is easy to construct three disks (Euclidean balls) such that
every two of them intersect but the intersection of the three is empty.
This implies that the Helly number of Euclidean geodesic balls
is strictly larger than two.
}


%
%

\section{The $L_1$ Geodesic Diameter}
\label{sec:diameter}
In this section, we show that the $L_1$ geodesic diameter of $P$, $\diam(P)$,
and a diametral pair can be computed in linear time
by extending the approach of Suri~\cite{s-agfnpsp-87} and
Hershberger and Suri~\cite{hs-msspm-97} to the $L_1$ case.
For any point $s\in P$, let $\far(s)$ be the maximum geodesic distance
from $s$ to any other point in $P$, that is, $\far(s) = \max_{q\in P} d(s, q)$.
A point $q\in P$ such that $d(s, q) = \far(s)$ is called
a \emph{farthest neighbor of $s$}.
Obviously, $\diam(P) = \max_{s\in P} \far(s)$ and $\rad(P) = \min_{s\in P} \far(s)$.
The following lemma is a key observation for our purpose.
\begin{lemma} \label{lem:farthest_neighbor}
 For any $s\in P$, all farthest neighbors of $s$ lie on the boundary $\bd P$ of $P$,
 and at least one of them is a vertex of $P$.
\end{lemma}
\newcommand{\pfFarNeig}{
\begin{proof}
We first show that any farthest neighbor of $s$ must lie on the boundary $\bd P$ of $P$.
Assume to the contrary that a point $q\in \intr P$
is a farthest neighbor of $s$.
Consider the Euclidean shortest path $\pi_2(s, q)$ from $s$ to $q$.
By Fact~\ref{fact:euc_simple}, $d(s, q) = |\pi_2(s, q)|$.
Let $q' \in \bd P$ be the point on $\bd P$
hit by the extension of the last segment of $\pi_2(s, q)$.
Observe that the Euclidean shortest path $\pi_2(s, q')$ from $s$ to $q'$
is obtained by concatenating $\pi_2(s, q)$ and the segment $\seg{qq'}$.
Since $q' \neq q$,
we have a strict inequality $d(s, q') = |\pi_2(s, q')| =|\pi_2(s, q)| + |\seg{qq'}| > d(s, q)$,
a contradiction to the assumption that $q$ is a farthest neighbor of $s$.
Therefore, there is no such farthest neighbor lying in the interior of $P$.
This proves the first statement of the lemma.

Let $q \in P$ be a farthest neighbor of $s$.
Then, $q$ lies on $\bd P$ as shown above.
Let $e$ be the edge of $P$ on which $q$ lies.
Corollary~\ref{coro:dist_seg} implies that
the geodesic distance $d(s, x)$ over $x \in e$ is quasiconvex
and thus it is maximized when $x$ is an endpoint of $e$,
that is, a vertex of $P$.
Thus, the lemma is shown.
\end{proof}
}
\ShoLong{}{\pfFarNeig}

\begin{corollary} \label{coro:diam_pair}
 There exist two vertices $v_1$ and $v_2$ of $P$
 such that $d(v_1, v_2) = \diam(P)$, that is,
 $(v_1, v_2)$ is a diametral pair.
\end{corollary}

Thus, the problem of computing $\diam(P)$ is solved by finding
the farthest vertex-pair.
Let $v_1, \ldots, v_n$ be the vertices of $P$ ordered counterclockwise along $\bd P$.
Let $v_a$ and $v_b$ be vertices of $P$ such that $v_a$ is a farthest neighbor of $v_1$ and
$v_b$ is a farthest neighbor of $v_a$.
The existence of $v_a$ and $v_b$ is guaranteed by Lemma~\ref{lem:farthest_neighbor}.
We assume that $a < b$; otherwise,
we take the mirror image of $P$ for the following discussion.
The three vertices $v_1, v_a, v_b$ divide $\bd P$ into three chains:
$U_1 = (v_2, \ldots, v_{a-1})$, $U_2 = (v_{a+1}, \ldots, v_{b-1})$,
and $U_3 = (v_{b+1},\ldots, v_n)$.
Let $W_1, W_2, W_3$ be the chains complementary to $U_1, U_2, U_3$, respectively,
that is, $W_1 = (v_a, \ldots, v_n, v_1)$, $W_2 = (v_b, \dots, v_n, v_1, \ldots, v_a)$,
and $W_3 = (v_1, \ldots, v_b)$.
We then observe the following, which we prove with Lemma~\ref{lem:farthest_neighbor}.
\begin{lemma}
\label{lem:three_chains}
 For any $i = 1,2,3$ and $u\in U_i$,
 there is a vertex $w\in W_i$ that is a farthest neighbor of $u$.
\end{lemma}
\newcommand{\pfThreeChains}{
\begin{proof}
We note that an almost identical statement (Lemma 8 in~\cite{s-cgfnsp-89}) was given by Suri for the Euclidean geodesic distance. In the following we give the details for completeness. Also note that it suffices to give the proof for $i=1$ only, since the other cases are analogous.

Let $u \in U_1$: by Lemma~\ref{lem:farthest_neighbor} there exists a vertex $v$ of $P$ that is a farthest neighbor of $u$.
If $v\in W_1$, we are done, hence from now on we focus in the $v \in U_1$ case. Walk along the boundary of $U_1$ in counterclockwise fashion and consider the order in which we meet points $v_1$, $u$, $v$, and $v_a$: the order is either $v_1, u, v, v_a$ or $v_1, v, u, v_a$. In the former case, the paths $\pi_2(v_1,v)$ and $\pi_2(u,v_a)$ must cross (since $u$ and $v_a$ lie on
 different sides of the path $\pi_2(v_1,v)$). In particular, by the triangular inequality we have $d(u, v_a) + d(v_1, v) \geq  d(u, v) + d(v_1, v_a)$. Recall that $v_a$ is a farthest neighbor of $v_1$, hence $d(v_1, v_a) \geq d(v_1, v)$, and in particular $d(u, v_a) \geq d(u, v)$. Thus, we conclude that $v_a\in W_1$ is also a farthest neighbor of $u$ as claimed. The case in which the order is $v_1, v, u, v_a$ is analogous (but in this case we find that $v_b$ is a farthest neighbor of $u$).


\end{proof}
}
\ShoLong{}{\pfThreeChains}

Lemma~\ref{lem:three_chains} implies that
computing a farthest vertex from every vertex of $P$
can be done by handling three pairs $(U_i, W_i)$ of two disjoint chains that
partition the vertices of $P$. We recall that a very similar strategy was used by Suri was also used for computing the Euclidean geodesic diameter~\cite{s-cgfnsp-89,hs-msspm-97}. This motivates the \emph{restricted farthest neighbor} problem:
Given two disjoint chains of vertices of $P$, $U = (u_1, \ldots, u_p)$ and
$W = (w_1, \ldots, w_m)$ that together partition the vertices of $P$,
where the vertices $u_1, \ldots, u_p,$ $w_1, \ldots, w_m$ are ordered counterclockwise
and $p+m = n$, find a farthest vertex on $W$ from each $u\in U$.
With respect to the Euclidean geodesic distance,
Suri~\cite{s-cgfnsp-89} presented an $O(n\log n)$-time algorithm for the problem,
and later Hershberger and Suri~\cite{hs-msspm-97} improved it to $O(n)$ time
based on the matrix searching technique by Aggarwal et al.~\cite{akmsw-gamsa-87}.
In the following, we show
with Fact~\ref{fact:euc_simple}
that the method of Hershberger and Suri~\cite{hs-msspm-97}
can be applied to solve the problem with respect to the $L_1$ geodesic distance $d$.
\begin{lemma}
\label{lem:farthest_vertex_prob}
 Let $U$ and $W$ be two disjoint chains of vertices of $P$ that together partition
 the vertices of $P$.
 One can compute in $O(n)$ time a farthest vertex over $w\in W$ for every $u\in U$
 with respect to the $L_1$ geodesic distance $d$.
\end{lemma}
\newcommand{\pfFarthestVertex}{
\begin{proof}
Define $M$ to be a $p \times m$ matrix such that
$M(i, j) := d(u_i, w_j)$ for $1\leq i \leq p$ and $1\leq j \leq m$.
We claim that $M(i,k)<M(i,l)$ implies
$M(j,k)<M(j,l)$ for any $1\leq i<j\leq p$ and $1\leq k<l\leq m$.
that is, $M$ is \emph{totally monotone}~\cite{akmsw-gamsa-87}.
Thus, the problem is to find the row-wise maxima in the totally monotone matrix $M$.

Indeed, the above claim follows from the fact that the counterclockwise order of the four vertices must be $w_k$, $w_l$, $u_i$, and $u_j$. In particular, the paths $\pi_2(u_i,w_k)$ and $\pi_2(u_j,w_l)$ must cross, and thus the triangular inequality implies $d(u_j,w_k)+d(u_i,w_l)\leq d(u_j,w_l)+d(u_i,w_k)$. This implies that we cannot have both $M(i,k)<M(i,l)$ and $M(j,l)<M(j,k)$ or we would have a contradiction. Note that this proof only uses the triangular inequality (for geodesic $L_1$ paths), hence it holds for other metrics as well. Indeed, this fact was previously shown by Hershberger and Suri~\cite{hs-msspm-97} for the Euclidean case.

Aggarwal et al.~\cite{akmsw-gamsa-87} proved that
the row-wise maxima of a totally monotone matrix
can be computed in $O(n)$ comparisons and evaluations of matrix entries.
The matrix $M$ is implicitly defined and each entry will be evaluated
only when needed.
Hershberger and Suri~\cite{hs-msspm-97} showed that
$O(n)$ evaluations of $M(i, j)$ can be done in total $O(n)$ time.
The main structures used in their algorithm are
the two \emph{shortest path trees} $\spt(w_1)$ and $\spt(w_m)$
rooted at the vertices $w_1$ and $w_m$, respectively,
and the funnel structures.
The shortest path tree $\spt(s)$ rooted at a point $s\in P$
is a plane tree on the vertices of $P$ plus the root $s$
such that the path in $\spt(s)$ from $s$ to any vertex $v$ of $P$
is actually the Euclidean shortest path $\pi_2(s, v)$.
It is known that the shortest path tree $\spt(s)$ can be computed
in linear time~\cite{ghlst-ltavspptsp-87}.

Fact~\ref{fact:euc_simple} implies that these structures can also be used
for the $L_1$ geodesic distance $d$ and the corresponding operations on them
can be performed in the same time bound without modifying the structures,
but by replacing the Euclidean geodesic distance by the $L_1$ geodesic distance $d$.
For example, we can use exactly the same
data structures as in~\cite{hs-msspm-97} for computing tangents on funnels.
Therefore, the algorithm of Hershberger and Suri applies to the $L_1$ case,
and solves the restricted farthest neighbor problem in linear time.
\end{proof}
}
\ShoLong{}{\pfFarthestVertex}

We are now ready to conclude this section with a linear-time algorithm.
We first find $v_a$ and $v_b$ such that
$v_a$ is a farthest neighbor of $v_1$ and $v_b$ is a farthest neighbor of $v_a$.
This can be done in $O(n)$ time by computing the shortest path maps
$\spm(v_1)$ and then $\spm(v_a)$ with the algorithm of Guibas et al.~\cite{ghlst-ltavspptsp-87}
and Fact~\ref{fact:euc_simple}.
We then have the three chains $U_1, U_2, U_3$ and their complements
$W_1, W_2, W_3$.
Next, we apply Lemma~\ref{lem:farthest_vertex_prob} to solve
the three instances $(U_i, W_i)$ for $i=1,2,3$ of the restricted farthest neighbor problem,
resulting in a farthest neighbor of each vertex of $P$ by Lemma~\ref{lem:three_chains}.
Corollary~\ref{coro:diam_pair} guarantees that the maximum over the $n$ pairs of vertices
is a diametral pair of $P$.
All the effort in the above algorithm is bounded by $O(n)$ time.
We finally conclude with the main result of this section.
\begin{theorem} \label{thm:diameter}
 The $L_1$ geodesic diameter of a simple polygon with $n$ vertices  can be computed in $O(n)$ time,
 along with a pair of vertices that is diametral.
\end{theorem}

\section{The $L_1$ Geodesic Center}
\label{sec:center}

In this section, we study the $L_1$ geodesic radius $\rad(P)$ and center
$\geocen(P)$ of a simple polygon $P$, and
present a simple algorithm that computes $\geocen(P)$ in linear time.


Consider the geodesic balls $\geoball_p(r)$ centered at all points $p\in P$ with radius $r$, and imagine  their intersection as $r$ grows continuously. By definition, the first nonempty intersection happens when $r = \rad(P)$.
Equivalently, by Theorem~\ref{thm:geoball_helly}, $r = \rad(P)$ is the smallest radius such that $\geoball_p(r) \cap \geoball_q(r) \neq \emptyset$ for any $p, q\in P$.


\begin{lemma}\label{lem:rad=diam/2}
 For any simple polygon $P$, it holds that $\rad(P) = \diam(P)/2$.
\end{lemma}
\newcommand{\pfRadDiam}{
\begin{proof}
Before giving the proof, we note that Schuierer~\cite{s-cl1dcsrp-94} claimed Lemma~\ref{lem:rad=diam/2} 
but no proof was given. In this paper, we provide a proof based on the Helly-type theorem for $L_1$ geodesic balls (Theorem~\ref{thm:geoball_helly}). It is worth mentioning that Theorem~\ref{thm:helly_simple} was also used to prove a similar relation between the diameter and center with respect to the rectilinear link distance~\cite{km-sbrlrsrp-11}.

By the triangle inequality we have $\rad(P) \geq \diam(P)/2$. Thus, it suffices to show the reverse direction. For any $r>0$, let $\mathcal{B}(r) := \{ \geoball_p(r) \mid p \in P\}$
be the family of $L_1$ geodesic balls with radius $r$.
Also, let $U(r) := \bigcap_{B \in \mathcal{B}(r)} B$ be their intersection.
As discussed above, $U(r) = \emptyset$ for $r<\rad(P)$
and $U(r) \neq \emptyset$ for $r\geq \rad(P)$.
Theorem~\ref{thm:geoball_helly}, together with the the above discussion,
tells us that any two members of $\mathcal{B}(r)$ have a nonempty intersection
only when $r \geq \rad(P)$.
This implies that for any $r<\rad(P)$
there are two points $p, q\in P$ such that $\geoball_p(r) \cap \geoball_q(r) = \emptyset$.

\end{proof}
}
\ShoLong{}{\pfRadDiam}

Clearly $\geocen(P) \subseteq \geoball_p(diam(P)/2) \cap \geoball_q(diam(P)/2)$ for any diametral pair $p,q\in P$. Since the intersection of these two balls can only be a segment, we obtain the following description of $\geocen(P)$.

\begin{corollary} \label{coro:geocen_shape}
 The $L_1$ geodesic center $\geocen(P)$ forms a line segment
 of slope $\pm 1$, unless it degenerates to a point.
\end{corollary}

Before explicitly computing $\geocen(P)$, we need a technical lemma. 

\begin{lemma} \label{lem:intersection_of_balls}
 Let $a, b \in P$ be any two points with $\seg{ab} \subset P$.
 Then, for any $r>0$, it holds that
  $\geoball_a(r) \cap \geoball_b(r) = \bigcap_{s\in \seg{ab}} \geoball_s(r)$.
\end{lemma}
\newcommand{\pfIntersectionBalls}{
\begin{proof}
Obviously, we have
$\bigcap_{s\in \seg{ab}} \geoball_s(r)\subseteq \geoball_a(r) \cap \geoball_b(r)$.
In the following, we show the opposite inclusion,
$\geoball_a(r) \cap \geoball_b(r) \subseteq \bigcap_{s\in \seg{ab}} \geoball_s(r)$. The statement is obviously true when $\geoball_a(r) \cap \geoball_b(r) =\emptyset$. Thus, from now on we assume that the intersection is nonempty. Pick any  point $p \in \geoball_a(r) \cap \geoball_b(r)$. Then, we have $d(p, a) \leq r$ and $d(p, b) \leq r$.
Consider the function $f(x) := d(p, x)$ over all points $x\in \seg{ab}$ on the segment $\seg{ab}$.
The function $f$ is quasiconvex by Corollary~\ref{coro:dist_seg},
and thus for any $x\in \seg{ab}$, $d(p,x) \leq \max\{d(p, a), d(p, b)\} \leq r$.
This implies that $p\in \geoball_x(r)$ for any $x\in \seg{ab}$,
so we conclude that $p\in \bigcap_{s\in \seg{ab}} \geoball_s(r)$.
\end{proof}
}
\ShoLong{}{\pfIntersectionBalls}


\subsection{Computing the center in linear time}
Now, we describe our algorithm for computing $\geocen(P)$ in linear time.
We start by computing the diameter $\diam(P)$ and a diametral pair of vertices $(v_1, v_2)$
in $O(n)$ time by Theorem~\ref{thm:diameter}.
Then, we know that $\rad(P) = \diam(P)/2$ by Lemma~\ref{lem:rad=diam/2}.
Compute the intersection of the two geodesic balls $\geoball_{v_1}(\rad(P))$ and
$\geoball_{v_2}(\rad(P))$, which is a line segment of slope $1$ or
$-1$ (including the degenerate case in which it is a single point). Extend this line segment to a diagonal $\ell=\seg{ab}$, where $a, b\in \bd P$. Clearly, we have $\geocen(P)\subseteq \seg{ab}$, but for simplicity in the exposition we will search for the possible centres of $P$ within $\ell$.

This preprocessing can be done in linear time: first, we must compute the balls $\geoball_{v_1}(\rad(P))$ and $\geoball_{v_2}(\rad(P))$ by computing the shortest path maps
$\spm(v_1)$ and $\spm(v_2)$ and traversing the cells of the maps. The intersection can be found by a local search at the midpoint of $\pi_2(v_1, v_2)$
since it is guaranteed that $\geoball_{v_1}(\rad(P)) \cap \geoball_{v_2}(\rad(P))$ is a line segment
by Corollary~\ref{coro:geoball_pconvex1}.

Since $\geocen(P) = \bigcap_{v \in V} \geoball_v(\rad(P))$
and $\geoball_{v_1}(\rad(P)) \cap \geoball_{v_2}(\rad(P))\subseteq \ell$,
we conclude that $\geocen(P) \subseteq \ell$.
The last task is thus to identify $\geocen(P)$ from $\ell$.
Here, we present a simple method based on further geometric observations.
Recall that for any $s\in P$ and any line segment $l\subset P$, the geodesic distance function $d(s, x)$ over $x\in l$ is quasiconvex as stated in Corollary~\ref{coro:dist_seg}.
A more careful analysis based on Fact~\ref{fact:l1length} gives us the following.
\begin{lemma} \label{lem:dist_seg_slope1}
 Given a point $s\in P$ and a line segment $\seg{ab} \subset P$ with slope $1$ or $-1$,
 let $f(x) = d(s, x)$ be the geodesic distance from $s$ to $x$ over $x\in \seg{ab}$.
 Then, there are two points $x_1, x_2\in \seg{ab}$ with $|\seg{ax_1}| \leq |\seg{ax_2}|$
 such that we have
 \[
    f(x) = \left\{\begin{array}{ll}
			d(s,a) - |\seg{ax}| & \text{if } x \in \seg{ax_1}\\
 			d(s,x_1) = d(s, x_2) & \text{if } x \in \seg{x_1x_2}\\
			d(s,b) - |\seg{bx}| & \text{if } x\in \seg{x_2b}.
		\end{array}\right.
  \]
 In particular, the function $f$ attains its minimum at any point $x\in \seg{x_1x_2}$.
\end{lemma}
\newcommand{\pfSegSlope}{
\begin{proof}
Consider the union of all Euclidean shortest paths $\pi_2(s, x)$ from $s$ to
all $x\in \seg{ab}$, which forms a funnel $F$ with apex $u$
and base $\seg{ab}$ plus $\pi_2(s, u)$.
Without loss of generality, we assume that the segment $\seg{ab}$
has slope $1$, $a$ is lower than $b$, and the apex $u$ of $F$ lies above
the line supporting $\seg{ab}$ (e.g., see \figurename~\ref{fig:fig3}).
The other cases are all symmetric to this situation.

\begin{figure}[tb]
  \center
  \includegraphics[width=.4\textwidth]{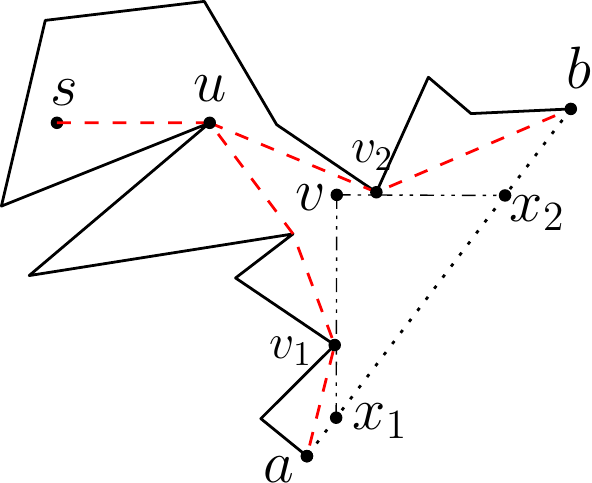}
  \caption{Illustration of the proof of Lemma~\ref{lem:dist_seg_slope1}:
  $C_1$ (resp., $C_2$) is the (red) chain connecting $u$ and $a$
  (resp., $b$). }
  \label{fig:fig3}
\end{figure}

Let $C_1$ and $C_2$ be the two concave chains of $F$,
where $C_1$ ends at $a$ and $C_2$ ends at $b$.
We define two particular vertices $v_1$ and $v_2$ on $C_1$ and $C_2$,
respectively, as follows. If we move on $C_1$ from $u$ to $a$, $v_1$
is the first point that is vertically visible to $\seg{ab}$.
If we move on $C_2$ from $u$ to $a$, $v_2$
is the first point that is horizontally visible to $\seg{ab}$.
We further define $x_1$ as the vertical projection of $v_1$ on $\seg{ab}$
and define $x_2$ as the horizontal projection of $v_2$ on $\seg{ab}$ (see \figurename~\ref{fig:fig3} for an illustration).

To prove the lemma, we use some results given in \cite{ciw-14}. First of all, we use the fact that $x_1$ is below
$x_2$. That is, point $x_1$ is contained in the line segment $\seg{ax_2}$  (Observation 1 of~\cite{ciw-14}).
Consider any point $x$ on the segment $\seg{ax_1}$. It has been proved
in~\cite{ciw-14} (Lemma 10) that there exists a shortest path from $s$ to $x$ in $P$ that passes through $x_1$.
This implies that $d(s,x)=d(s,x_1)+|\seg{x_1x}|$, and in particular it holds that $d(s,x)=d(s,a)-\seg{ax}$ (for any $x\in\seg{ax_1}$).
Similarly, the same holds for any point $x$ on the segment $\seg{bx_2}$.

It remains to to prove that $d(s,x)=d(s,x_1)=d(s,x_2)$ for any $x\in
\seg{x_1x_2}$. Let $v$ be the intersection of the vertical line
containing $\seg{v_1x_1}$ and the horizontal line containing
$\seg{v_2x_2}$ (e.g., see \figurename~\ref{fig:fig3}). It has been proved in
 \cite{ciw-14} (Lemma 11) that $v$ must be in the funnel $F$ (in fact,
 the entire triangle $\triangle vx_1x_2$ is in
 $F$) and for any point $x\in \seg{x_1x_2}$ there must be a
 shortest path from $s$ to $x$ in $P$ that passes through $v$. Thus, we have $d(s,x)=d(s,v)+|\seg{vx}|$ for any $x\in \seg{x_1x_2}$. Since $\seg{ab}$ has slope $1$, the value $|\seg{vx}|$ is constant for all $x\in \seg{x_1x_2}$. Thus, we obtain $d(s,x)=d(s,x_1)=d(s,x_2)$ for any $x\in \seg{x_1x_2}$, as claimed.
\end{proof}
}
\ShoLong{}{\pfSegSlope}

For any vertex $v \in V$ of $P$, let $\ell_v \subseteq \ell$ be
the intersection $\geoball_v(\rad(P)) \cap \ell$.
Since $\geocen(P) = \bigcap_{v \in V} \geoball_v(\rad(P))$ and
$\geocen(P)\subseteq \ell$,
it holds that $\geocen(P) = \bigcap_{v \in V} \ell_v$.

\begin{lemma}\label{lem:ell_v}
 For any vertex $v$ of $P$, $\ell_v$ can be computed
in $O(1)$ time, provided that $d(v, a)$ and $d(v, b)$ have been evaluated, where $a$ and $b$ are the endpoints of $\ell$.
\end{lemma}
\newcommand{\pfEllV}{
\begin{proof}
Let $g_v(x) := d(v, x)$ over $x\in \ell$.
By Lemma~\ref{lem:dist_seg_slope1},
the graph of $g_v$ consists of at most three line segments
whose slopes are $-1$, $0$, and $1$ in order as $|\seg{ax}|$ increases.
Let $x_1, x_2\in \ell = \seg{ab}$ be the breakpoints of $g_v$ as defined in
Lemma~\ref{lem:dist_seg_slope1}.
\ShoLong{

}{See \figurename~\ref{fig:ell_v} for an illustration.

\begin{figure}[t]
  \center
  \includegraphics[width=.5\textwidth]{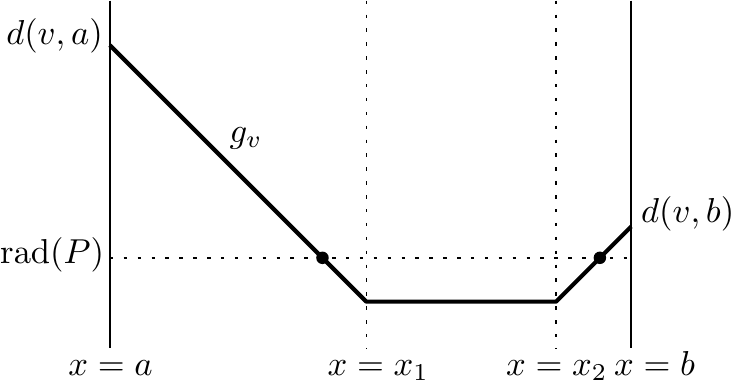}
  \caption{Illustration of the proof of Lemma~\ref{lem:ell_v}. }
  \label{fig:ell_v}
  \vspace{-10pt}
\end{figure}
}
Note that $\ell_v$ is the sublevel set of $g_v$ at $\rad(P)$,
that is, $\ell_v = \{ x \in \ell \mid g_v(x) \leq \rad(P)\}$.
Since $\geocen(P)\neq\emptyset$ and thus $\ell_v \neq \emptyset$,
no endpoint of $\ell_v$ lies strictly between $x_1$ and $x_2$.
In other words, both endpoints of $\ell_v$ should lie in the subset
$\seg{ax_1} \cup \seg{x_2b}$.
Thus, once we know the values of $g_v(a) = d(v, a)$ and $g_v(b) = d(v, b)$,
we can identify the endpoints of $\ell_v$ in $O(1)$ time
by Lemma~\ref{lem:dist_seg_slope1}.
\end{proof}
}
\ShoLong{}{\pfEllV}

Thus, our last task can be completed as follows:
Compute the two shortest path maps $\spm(a)$ and $\spm(b)$ with sources $a$ and $b$, respectively,
by running the algorithm by Guibas et al.~\cite{ghlst-ltavspptsp-87}.
This evaluates $d(v, a)$ and $d(v, b)$ for all vertices $v$ of $P$ in linear time.
Next, we compute $\ell_v$ for all vertices $v$ of $P$ by Lemma~\ref{lem:ell_v}
and find their common intersection, which finally identifies $\geocen(P)$.
All the effort to obtain $\geocen(P)$ is bounded by $O(n)$.

\begin{theorem} \label{thm:center}
 The $L_1$ geodesic radius and all centers of a simple polygon with $n$ vertices
 can be computed in $O(n)$ time.
\end{theorem}

\section{Concluding Remarks}
\label{sec:conclusion}

In this paper, we presented a comprehensive study on the
$L_1$ geodesic diameter and center of simple polygons,
resulting in optimal linear-time algorithms.
Our approach relies on observations about $L_1$ geodesic balls,
in particular, the $P$-convexity (Lemma~\ref{lem:geoball_pconvex}) and
the Helly-type theorem (Theorem~\ref{thm:geoball_helly}).
These are key tools to show structural properties of
the diameter and center.

One might be interested in extending this framework to polygons with holes,
namely, \emph{polygonal domains}.
However, it is not difficult to see that few of the observations we made
extend to general polygonal domains.
First and foremost, an $L_1$ (also, Euclidean) geodesic ball may not be
$P$-convex when $P$ has a hole.
In addition, the Helly number of $L_1$ geodesic balls in a polygonal domain
is strictly larger than two:
one can easily construct three balls around a hole such that
every two of them intersect but the three have no common point.
Also, Lemma~\ref{lem:farthest_neighbor} (the existence
of a farthest neighbor that is a vertex) does not always hold in polygonal domains.
Bae et al.~\cite{bko-gdpd-13} have exhibited several examples
of polygonal domains in which a farthest neighbor with respect to the Euclidean geodesic distance
is a unique point in the interior. This construction can easily be extended to the $L_1$ geodesic distance.

\section*{Acknowledgements}
S.W.~Bae was supported by Basic Science Research Program through the National Research Foundation of Korea (NRF) funded by the Ministry of Science, ICT \& Future Planning (2013R1A1A1A05006927). Y.~Okamoto was supported by Grant-in-Aid for Scientific Research from Ministry of Education, Science and Culture, Japan, and Japan Society for the Promotion of Science (JSPS). H. Wang was supported in part by NSF under Grant CCF-1317143.



{
\bibliographystyle{splncs03}
\bibliography{geodiamcen}

}

%
\end{document}